\documentclass[journal]{IEEEtran}
\usepackage[utf8]{inputenc}
\usepackage[colorlinks=true,linkcolor=blue,citecolor=blue]{hyperref}
\usepackage[dvipsnames]{xcolor}
\usepackage{booktabs}
\usepackage{tikz}

\usepackage{amsmath, amssymb, amsfonts, amsthm, mathtools, bm}
\usepackage{physics}
\usepackage{thmtools, thm-restate}
\usepackage{cancel}
\usepackage[bb=boondox,bbscaled=.95]{mathalfa}

\newtheorem{defn}{Definition}
\newtheorem{rmk}{Remark}

\newtheorem{lemma}[defn]{Lemma}

\DeclareMathOperator{\Span}{span}
\DeclareMathOperator{\vect}{vec}
\DeclareMathOperator{\diag}{diag}
\DeclareMathOperator{\rad}{rad}

\newcommand{\nA}{{n_a}}

\newcommand{\nH}{{n_h}}
\newcommand{\nQ}{{n_q}}
\newcommand{\nT}{{n_\theta}}
\newcommand{\nU}{{n_u}}
\newcommand{\nX}{{n_x}}
\newcommand{\nPP}{{n_{\theta,\phi}}}
\newcommand{\nPK}{{n_{\theta,K}}}
\newcommand{\pphi}{{\theta_\phi}}
\newcommand{\pK}{{\theta_K}}
\newcommand{\Th}{{\theta}}

\newcommand{\cB}{\mathcal{B}}
\newcommand{\cC}{\mathcal{C}}
\newcommand{\cF}{\mathcal{F}}
\newcommand{\Ha}{\mathcal{H}}

\newcommand{\cM}{\mathcal{M}}
\newcommand{\cO}{\mathcal{O}}
\newcommand{\cQ}{\mathcal{Q}}
\newcommand{\cT}{\mathcal{T}}
\newcommand{\cU}{\mathcal{U}}
\newcommand{\cX}{\mathcal{X}}
\newcommand{\cY}{\mathcal{Y}}

\newcommand{\bz}{\mathbb{0}}

\newcommand{\R}{\mathbb{R}}

\newcommand{\bP}{\mathbb{P}}

\definecolor{olive}{rgb}{0.6, 0.6, 0.2}
\definecolor{sand}{rgb}{0.8666666666666667, 0.8, 0.4666666666666667}
\definecolor{wine}{rgb}{0.5333333333333333, 0.13333333333333333, 0.3333333333333333}
\definecolor{deblue}{RGB}{11,132,147}
\definecolor{ocra}{RGB}{204, 119, 34}

\newcommand{\fcircle}[2][red,fill=red]{\tikz[baseline=-0.5ex]\draw[#1,radius=#2] (0,0.03) circle ;}
\begin{document}
\title{Optimal Energy Shaping via Neural Approximators}
\author{Stefano Massaroli$^{1,\star}$, Michael Poli$^{2, \star}$, Federico Califano$^{3, \star}$,\\ Jinkyoo Park$^{2}$, Atsushi Yamashita$^{1}$, Hajime Asama$^{1}$
\thanks{
$^\star$Equal contribution authors. $^1$Department of Precision Engineering, The University of Tokyo. $^2$ Department of Industrial Engineering, Korean Advanced Institute of Science and Technology (KAIST). $^{3}$Faculty of Electrical Engineering, Mathematics \& Computer Science, University of Twente. 
\newline
Contacts: {\tt massaroli@robot.t.u-tokyo.ac.jp},\newline
{\tt poli\textunderscore m@kaist.ac.kr}, {\tt f.califano@utwente.nl}
}}
\maketitle

 \begin{abstract}
We introduce \textit{optimal energy shaping} as an enhancement of classical passivity--based control methods. A promising feature of passivity theory, alongside stability, has traditionally been claimed to be intuitive performance tuning along the execution of a given task. However, a systematic approach to adjust performance within a passive control framework has yet to be developed, as each method relies on few and problem--specific practical insights. Here, we cast the classic energy--shaping control design process in an optimal control framework;
once a task--dependent performance metric is defined, an optimal solution is systematically obtained through an iterative procedure relying on neural networks and gradient--based optimization. The proposed method is validated on state--regulation tasks.

\end{abstract}
%

\section{Introduction}%
%
The concept of \textit{energy} experienced, in the last three decades, a surge of attention in both engineering practice as well as robot control theory. \textit{Passivity--based control} (PBC) has since fully established itself as a branch of nonlinear control. Its core idea is to treat dynamical systems as physical entities able to exchange energy with other systems rather than to process signals \cite{van2000l2, ortega2001putting}. The resulting framework, called \textit{port-Hamiltonian} theory, links the concept of energy and the input/output characterization of systems \cite{duindam2009modeling}.
The methodology grew fast, and passivity--based control methods are now ubiquitous in robotics: from {gravity compensation} \cite{80014492902} or {impedance control} \cite{stramigioli2001modeling} to teleoperation \cite{secchi2007control}, formation control \cite{vos2014formation} or non--smooth tasks \cite{massaroliiterative}.
The core motivation of choosing passive control strategies, besides obtaining stability of the closed-loop system, is to guarantee a stable interaction with an unknown passive environment \cite{Stramigioli2015EnergyAwareR}. 
{
In fact the latter, especially when manifests with stiff impedance parameters, could drastically modify the dynamic evolution of the systems, leading to unsafe behaviours.
}
In robotics, this framework has evolved into the so called \textit{energy aware robotics} methodology, where safety constraints and energy budgeting techniques can be consistently integrated in the \textit{intrinsically passive control} architecture \cite{Stramigioli2015EnergyAwareR}. A limit of this approach is that the design of the low--level controller is dedicated to achieving passivity of the closed--loop system, without optimizing for task performance. This aspect produced the highly--debated assertion that passivity is associated to a loss of performance in the control due to the choice of over--conservative gains \cite{Dimeas2016}. This loss of performance can be partially addressed by imposing task--dependent \textit{energy budgeting} protocols (see e.g., \cite{groothuis2018general} and references therein) at a higher supervision level of the controller. Nevertheless the often claimed intuitive performance tuning of energy--based control methods is questionable, since in most cases the developed ideas were limited to a re--derivation of PD controllers using novel arguments, \textit{de facto} obscuring the power of the technique to the community.

In most physical domains the predominant class of passivity--based control methods for physical systems is represented by \textit{potential compensation} techniques; these rely on canceling the system's potential part of the energy and replace it with a fictitious quadratic term having the minimum placed in the desired set point. Then, dissipation is often added to the system trough a (fictitious) linear damper. Although this approach presents potential issues related to unstable nonlinear terms introduced by imperfect cancellation of the system's potential (e.g., in case of uncertainties on model's parameters), it remains very popular among, e.g., the robot control community. Furthermore, besides the problem of tuning the gains of the proportional and derivative actions, two deeper questions remain unanswered: \textit{why should we assign a quadratic potential?} and \textit{why should we choose a linear dissipation term?}

For what concerns performance tuning, pure learning--based approaches to control such as reinforcement learning methods \cite{bertsekas1995dynamic,sutton2018reinforcement}, generally place utmost importance on performance optimization at the expense of stability and safety. In particular, stability is only rarely ensured through \textit{ad hoc} algorithmic extensions rather than representing an intrinsic property of the framework. The limitations of this \textit{modus operandi}, permeating the modern landscape of neural network learning--based control and prediction, have been highlighted by the community \cite{bucsoniu2018reinforcement,dulac2019challenges}, and several approaches have been proposed since \cite{berkenkamp2017safe}. 

Here, instead of an attempt to patch out weaknesses starting from a pure learning framework, we propose carefully blending the expressive power of neural networks and passivity based control. To achieve this goal, we systematically design a parametrized passive controller through the solution of an optimization problem, in order to maximize the performance of the system along a task execution. 
A formulation of the optimal control problem involving energy--balancing passive controllers is then obtained through an extension of the continuous--depth framework, recently developed in machine learning. This paradigm, including one of its core model classes, the \textit{Neural Ordinary Differential Equation} (Neural ODE) \cite{chen2018neural, massaroli2020dissecting} has been successfully employed in combination to energy based models \cite{greydanus2019hamiltonian, lutter2019deep, cranmer2020lagrangian, massaroli2020stable}, offering a grounded way to introduce first--principles into neural networks.
\newpage
\subsubsection*{Notation}%
Let $\cM$ be a smooth manifold. For all $x\in\cM$, $T_{x}\cM$ and $T_{x}^*\cM$ are the tangent and cotangent spaces in $x$, while $T\cM$ and $T^*\cM$ are its tangent and cotangent bundles, respectively. Moreover, let us denote with $\prec$ ($\preceq$) and $\succ$ ($\succeq$) negative and positive (semi--) definite matrices. Given a vector space $V$, we devote with $\langle\cdot, \cdot\rangle:V\times V^*\rightarrow\R$ the natural duality paring. $\bz_n$ represents the $n-$dimensional zero vector.
Let $\Ha(x):\cM\rightarrow\R$ be a scalar function over a smooth manifold $\cM$. ${\partial\Ha}/{\partial x}\in T^*\cM$ is the gradient of $\Ha$ and is represented in coordinates as a row vector. We refer to $\nabla\Ha\in T\cM$ as the vector gradient, which in a coordinate representation is a column vector, the transposed of the aforementioned gradient. We use the triplet $(\Omega, \cF, \bP)$ to denote a \textit{probability space} with sample space $\Omega$, $\sigma-$algebra $\cF$ and measure $\bP$. Given a random variable $\alpha$ defined on such probability space, we represent its expectation as $\mathbb{E}_{\omega\sim\bP(\omega)}[\alpha] :=\int_{\Omega}\alpha(\omega)\dd\bP(\omega)$.
%

\section{Background}%
In this section, we give an overview on the port--Hamiltonian representation of dynamics and the \textit{energy balancing} paradigm.
\subsection{Input--State--Output Model of Port--Hamiltonian Systems}
The input--state--output representation of a port--Hamiltonian system \cite{van2000l2} consists in a dynamic model of the form: 
\begin{equation}\label{eq:ph}
    \begin{matrix*}[l]
        &\dot x = F(x)\nabla\Ha(x) + g(x)u\\
        &y = g^\top(x)\nabla\Ha(x)
    \end{matrix*}
\end{equation}
with state $x\in\cX$ ($\cX$ is a $\nX$--dimensional manifold), input $u\in\cU$ ($\cU:=\cC^\infty(\R\rightarrow\R^\nU)$) and output $y\in\cY$. The system matrices are such that $F(x)+F^\top(x)\preceq 0$ ($F:\cX\rightarrow\R^{\nX\times\nX}$), input matrix $g:\cX\rightarrow\R^{\nX\times\nU}$ and Hamiltonian $\Ha:\cX\rightarrow\R$.

The main property of system (\ref{eq:ph}) is to be passive \cite{brynes1991} with respect to the storage function $\Ha$. In fact,
\begin{equation}
    \begin{aligned}
        \dot\Ha(x) &= \langle\nabla\Ha(x),\dot{x}\rangle\\
                &=\langle\nabla\Ha(x) ,F(x)\nabla\Ha(x)\rangle + \langle\nabla\Ha(x),g(x)u\rangle\\
                &=\underbrace{\langle\nabla\Ha(x) ,F(x))\nabla\Ha(x)\rangle}_{\leq 0} + \langle y,u\rangle\leq \langle y,u\rangle
    \end{aligned}
\end{equation}
Mechanical systems admits a representation in form (\ref{eq:ph}). Let $\cQ$ be a $\nQ$--dimensional configuration manifold and let $q\in\cQ$ be the vector of generalized coordinates. $p\in T_{q}^*\cQ$ denotes the generalized momenta, $p:=M(q)\dot q$, where $M(q)=M^\top(q)\succ 0~(M:\cQ\rightarrow\R^{\nQ\times\nQ})$ is the inertia matrix of the system. Note that, here the full state $x := (q, p)$ of the system is an element of the cotangent bundle $T^*\cQ$ of the configuration manifold $\cQ$. The equations of motions in canonical form are given by \eqref{eq:ph} with 
$$
    F := \begin{bmatrix}
            \mathbb{O} & \mathbb{I}\\
            -\mathbb{I} & D
        \end{bmatrix}, ~g :=\begin{bmatrix}
            \mathbb{O}\\
            B
        \end{bmatrix}
$$
resulting in 
\begin{equation}\label{eq:mech_ph}
    \begin{aligned}
        \begin{bmatrix}
            \dot q\\
            \dot p
        \end{bmatrix} &= 
        \begin{bmatrix}
            \mathbb{O} & \mathbb{I}\\
            -\mathbb{I} & D
        \end{bmatrix}
        \begin{bmatrix}
            \nabla_{q}\Ha(q, p)\\
            \nabla_{p}\Ha(q, p)
        \end{bmatrix} +
        \begin{bmatrix}
            \mathbb{O}\\
            B
\        \end{bmatrix}u\\
        y &= B\dot q
    \end{aligned}
\end{equation}
where $\Ha:T^*\cQ\rightarrow\R$ is the total energy (Hamiltonian)
\[
    \Ha(q,p):= \frac{1}{2}p^\top M^{-1}(q)p + V(q),
\]
$V:\cQ\rightarrow\R$ the \textit{potential}, $D=D^\top\preceq 0$ takes into account dissipative and friction effects, $B\in\R^{d\times d}$ is the input matrix and $u:\R\rightarrow\R^\nQ$ are the external generalized forces. Here, $q, p\in\R^\nQ, ~x\in\R^\nX, ~\nX=2\nQ$ in coordinates. Note that the rank of $B$ determines the actuation regime of the mechanical system. If $\rank(B)=\nQ$, the system is \textit{fully--actuated} while, if $\rank(B)<\nQ$ the system is \textit{under--actuated}.
\subsection{Passivity Based Control (PBC)}%
Before defining the proposed framework of \textit{optimal energy shaping}, we briefly introduce the classical setting.
\subsubsection{Energy Shaping}
Consider a PH system \eqref{eq:ph}; a control action $u = \beta(x) + v$ solves the PBC problem if the closed-loop system satisfies a desired power-balance equation
\begin{equation*}
    \dot\Ha^*(x) = \langle z(t), v(t)\rangle - d^*(t)
\end{equation*}
where $\Ha^*(x)$ is the desired energy function, $d^*(t)$ the desired dissipation function and $z\in\R^\nU$ the new power conjugated (passive) output.
The most straightforward solution of the energy shaping problem is the so called \textit{energy-balancing PBC} (EB-PBC). The controller can be obtained directly from the power balance equation by setting the desired dissipation $d^*$ equal to the natural dissipation of the system while retaining $y$ as the new passive output, i.e 
\[
d^* := -\langle\nabla\Ha, F(x)\nabla\Ha\rangle,~~z: = y.
\]
\begin{restatable}[EB--PBC \cite{ortega2008control}]{proposition}{propone}\label{prop:1}
If it is possible to find a function $\beta(x)$ such that $\dot \Ha^*(x) -\dot\Ha(x) = \langle y(t),\beta(x)\rangle$ then the control law $u = \beta(x)+v$ is such that  $\dot \Ha^*(x) = \langle y(t),v(t)\rangle -d^*$ is satisfied.
\end{restatable}
This means that the state feedback $\beta(x)$ is such that the \textit{added energy} $\Ha^*(x) -\Ha(x)$ equals the energy supplied to the system through the power port by $\beta$ and, consequently, $\Ha^*$ is the difference between the stored and supplied energy.
The closed-form solution of the EB-PBC controller is given by 
\[
    \beta(x) = -g^+(x)F^\top(x)\left(\nabla\Ha^*(x)-\nabla\Ha(x)\right)
\]
where $g^+(x)$ is the left pseudo--inverse of $g(x)$, $g(x):=(g^\top(x)g(x))^{-1} g^\top(x)$ and $\Ha^*(x)$ satisfies the following matching equations 
\begin{equation}\label{eq:matching_cond}
    \begin{bmatrix}
        g^{\perp}(x)F^\top(x)\\
        g^\top(x)
    \end{bmatrix}
    \left(\nabla\Ha^*(x)-\nabla\Ha(x)\right) = \mathbb{0}_{\nX+\nU}
\end{equation}
being $g^{\perp}(x)$ is a left full--rank annihilator of $g(x)$, i.e. $g^\perp(x)$ is such that for all $x\in\cX$, $g^\perp(x) g(x) = \mathbb{0}_\nX$.
\subsubsection{Damping Injection}
From standard passivity--based control \cite{ortega2008control} dissipation can be added to the system via negative output feedback as a mean to extract energy from the system. Let $K = K^\top\succ 0,~(K\in\R^{\nU\times\nU})$. Then, the output feedback
\[
    v = -K y 
\]
preserves passivity by adding dissipation to the system. Consider a port--Hamiltonian system controlled via energy shaping and let us add the damping injection term: $u:= \beta(x) - Ky$. Then, from Proposition \ref{prop:1} we have
\begin{equation}\label{eq:damping_inj}
    \begin{aligned}
        \dot\Ha^* &= \langle y(t),-Ky(t)\rangle - \langle\nabla\Ha(x),F(x)\nabla\Ha(x))\rangle \\
        &= -\nabla^\top\Ha(x) g(x)Kg^\top(x)\nabla\Ha(x) - \nabla^\top\Ha(x)F(x))\nabla\Ha(x)\\
        &= -\nabla^\top\Ha(x)\left[g(x)Kg^\top(x) + F(x))\right]\nabla\Ha(x)\leq 0
    \end{aligned}
\end{equation}
%

\section{Optimal Energy Balancing}
In general, energy--balancing PBC admits infinite possible choices of $\Ha^*(x)$ and $K$ for a given control problem. Nonetheless, in practice, quadratic desired energies and constant damping matrices dominate almost most approaches. 

Consider a time horizon $\cT:=[0, T]$ ($T>0$). Once a scalar \textit{cost function} $\ell$ is defined, we aim at finding desired energy $\Ha^*(x)$ and output feedback gain $K^*$ minimizing that cost. Note that we consider $K^*$ to be a a function of both time and the state, $K^*:\cT\times\cX\rightarrow\R^{\nU\times\nU},~K^* = K^*(t, x)$. In fact, as long as $K^*(t, x(t))\succeq 0$ for all $t\in\cT$, the output feedback law $v=-K^*(t, x)y$ will still extract energy from the system, as it can be seen by inspecting \eqref{eq:damping_inj}. 

If $F, \Ha, g$ and $u$ are smooth enough, for any initial condition $x_0\in\cX$ the port--Hamiltonian system \eqref{eq:ph} admits a unique solution defined on $\cT$. If this is the case, there exists a mapping $\Phi$ from $\cX\times\cU$ to the space of absolutely continuous functions $\cT\mapsto\cX$ such that $x(t) = \Phi(x_0, u)(t)$. For compactness we denote $\Phi(x_0, u)(t)$ by $\Phi_t(x_0, u)$.
In general, $\ell$ comprises two terms, a \textit{terminal cost}, depending on the final state of the system (at the end of the horizon) and an \textit{integral cost}, distributed on the whole time horizon, i.e.
\[
    \begin{aligned}
        \ell(x_0, u) &= \int_0^T l\left(t, u(t),\Phi_t(x_0, u)\right)\dd t\\
        &~~~+ L\left(T, u(T), \Phi_T(x_0, u)\right).
    \end{aligned}
\]
Moreover, let the system initial condition $x(0)=x_0$ be a random variable defined on a probability space $(\Omega, \cF, \bP)$ such that $x_0\in\cX$ almost surely ($\bP(x_0\in\cX)=1$). Namely, we suppose that only the distribution of initial conditions is known, and therefore, we wish to minimize the cost expectation over this distribution. 
This approach allows us to obtain optimal stable controllers for a whole distribution of initial conditions rather than a single one.
The optimal energy--balancing control design can be cast as the following nonlinear program:
\begin{equation}\label{eq:NP}
        \begin{aligned}
            \min_{\Ha^*,~K^*}  &\mathbb{E}_{x_0\sim \bP(x_0)}\left[\ell(x_0, u)\right]\\
            \text{subject to}~~
            &\dot x =  F(x)\nabla\Ha(x) + g(x)u(t,x)~~t\in [0,T]\\
            & u(t, x)= g^+(x)F^\top(x)\left(\nabla\Ha^*(x)- \nabla\Ha(x)\right)\\
            &\qquad~~-K^*(t, x)g^\top(x)\nabla\Ha(x)\\
            &\begin{bmatrix}
                g^{\perp}(x)F^\top(x)\\
                g^\top(x)
            \end{bmatrix}\left(\nabla\Ha^*(x)-\nabla\Ha(x)\right) = \mathbb{0}\\
            &K^*(t, x)\succ 0
        \end{aligned} 
\end{equation}
The above problem does not present, in general, any closed--form solution. Therefore, we hereby propose a method to reach an optimal solution by iterating a stochastic gradient descent \cite{robbins1951stochastic} algorithm. 
We start by breaking down (\ref{eq:NP}) in sub--problems, listing the issues needed to be overcome:
\begin{itemize}
    \item[1.] The nonlinear program presents a PDE constraint
    $$
        \begin{bmatrix}
                    g^{\perp}F^\top\\
                    g^\top
        \end{bmatrix}\left(\nabla\Ha^*(x)-\nabla\Ha(x)\right) = \mathbb{0}_{\nX+\nU};
    $$
    \item[2.] $\Ha^*(x)$ and $K^*(t, x)$ belong to some functional space, i.e. the optimization problem is infinite-dimensional in nature;
\end{itemize}
Moreover, our aim is to be as agnostic as possible with respect to the cost functional $\ell$ (which represents the task to be executed) while deriving a solution algorithm.
In the following, we address the mentioned problems and derive a systematic procedure to design the optimal energy balancing controller for a given system.

In order to reduce the complexity of the optimization, our first aim is to get rid of the PDE constraints. We therefore seek to characterize a subset of optimal solutions for which those constraints are automatically satisfied.
\begin{restatable}[Reduced--order energy shaping]{proposition}{PotShap}\label{prop:PDE_general}
    If there exists a map $\Lambda:\cX\rightarrow\R^{\nA\times\nA}$ ($\nA<\nX$) such that for any $x\in\cX$
    \begin{equation}
        \begin{aligned}
            &\rank(\Lambda(x)) = \nA\\
            &\ker\left(
                \begin{bmatrix}
                    g^\perp(x) F^\top(x)\\
                    g^\top(x)
                \end{bmatrix}\right)
            =
            \Span\left(\begin{bmatrix}
                    \Lambda(x)\\
                    \mathbb{O}
                \end{bmatrix}\right)
        \end{aligned}
    \end{equation}
    then, any desired energy function $\Ha^*(x) := \Ha(x) + \phi^*(x_a)$ with $x_a\in\cX_a$ ($\cX_a$ is a $n_a$--dimensional submanifold of $\cX$), $\phi^*:\cX_a\rightarrow\R$, $x = (x_a,x_b)$ ($x_b\in\cX\setminus\cX_a$), satisfies the matching equations (\ref{eq:matching_cond}).
\end{restatable}
Thus, according to to Proposition \ref{prop:PDE_general}, the nonlinear program (\ref{eq:NP}) can then be reformulated without the PDE constraint as
\begin{equation}\label{eq:NPunc}
        \begin{aligned}
            \min_{\phi^*,~K^*}  & \mathbb{E}_{x_0\sim\bP(x_0)}\left[\ell(x_0, u)\right]\\
            \text{subject to}~~
            &\dot x =  F(x)\nabla\Ha(x) + g(x)u(t,x)~~t\in [0,T]\\
            & u(t, x) = -g^+(x)F^\top(x)\nabla_x\phi^*(x_a)\\
            &\qquad~~~-K^*(t, x)g^\top(x)\nabla\Ha(x)\\
            &K^*(t, x)\succ 0\\
        \end{aligned}
\end{equation}
\begin{rmk}
From a control-theoretic point of view the introduced optimisation problem, which will be solved by means of neural approximators, provides a link between classical Lyapunov stability and performance of a dynamical system. The system \textit{learns} the optimal function $\phi^*$, and as consequence the optimal energy-like function $\mathcal{H}^*$ which is, by construction, a Lyapunov function for the closed-loop system whenever $\phi^*$ is bounded. Unlike most of the control strategies learned by neural networks, passivity (and as consequence stability) of the closed-loop system is a trivial corollary in this energy-based framework, in which the learning part of the algorithm acts without breaking down the structure of the underlying physical system.
\end{rmk}%

\section{Solving the Optimization Problem}
If Proposition \ref{prop:PDE_general} solves the issue about PDE constraints, in the following we address the second problem, concerning the infinite-dimensional nature of the space in which the decision variables live. In particular we restrain the search to an optimal solution by discretizing the problem by means of neural--networks.

\subsection{Stochastic Gradient Descent}
In this work the optimal energy shaping controller will be derived via \textit{stochastic gradient descent} (SGD) optimization \cite{robbins1951stochastic}. The algorithm works as follows: given a scalar function $\ell(\alpha, x_0)$ dependent on the variable $\alpha\in\R^{n_\alpha}$, it attempts at computing $\alpha^*=\arg \min_{\alpha} \mathbb{E}_{x_0\sim\bP(x_0)}[\ell(\alpha, x_0)]$ starting from an initial guess $\alpha_0$ by updating a candidate solution $\alpha_k$ recursively as follows:
\[
    \alpha_{k+1} = \alpha_{k} - \frac{\eta_k}{N}\sum\limits_{i=1}^N\frac{\dd}{\dd\alpha} \ell(\alpha_k, x_0^i)
\]
where $N$ is a predetermined number of samples $x_0^i$ of the random variable $x_0$ and $\eta_k$ is a positive scalar \textit{learning rate}. If $\eta_k$ is suitably chosen and $\ell$ is convex, $\alpha_k$ converges in expectation to the minimizer of $\ell$ as $k\rightarrow\infty$. Although global convergence is no longer guaranteed in the non--convex case, SGD--based techniques are widely used in practical applications, especially among the machine learning community, due to their scalability and computational efficiency.
\subsection{Finite dimensional optimization via neural approximators}
Let us assume $\phi^*\in\cC^\infty(\cX_a\rightarrow\R)$, $ K^*\in\cC^\infty(\cT\times\cX\rightarrow\R^{\nU\times\nU})$. After getting rid of the PDE constraints a solution of the nonlinear program (\ref{eq:NPunc}) could be obtained, in principle, by iterating the SGD algorithm in $\cC^\infty(\cX_a\rightarrow\R)\times\cC^\infty(\cT\times\cX\rightarrow\R^{\nU\times\nU})$ \cite{smyrlis2004local}. 

Let $\phi_k^*$, $K^*_k$ be the optimal solutions at the $k$th SGD iteration. The update rule of ${\phi}_k^*$ and $K^*_k$ is 
\begin{equation}
    \begin{aligned}
        \phi^*_{k+1}(x_a) &=  \phi^*_{k}(x_a) - \frac{\eta_k}{N}\sum\limits_{i=1}^N\dfrac{\delta}{\delta\phi^*}\ell(x_0^i, u)\\
         K^*_{k+1}(t,x) &=  K^*_{k}(t, x) - \frac{\eta_k}{N}\sum\limits_{i=1}^N\dfrac{\delta}{\delta K^*}\ell(x_0^i, u)
    \end{aligned}
\end{equation}
where $\delta/\delta (\cdot)$ is the variational or functional derivative. While local convergence to optimal solutions is still ensured under mild regularity and smoothness assumptions, obtaining derivatives in function space turns out to be computationally unfeasible. Additionally, besides of symbolic approaches, storing and manipulating functions in software implementation is impossible. We therefore seek to reduce the problem to finite dimension by approximating $\phi^*(x_a)$ and $K^*(t, x)$ with \textit{neural networks}, and show that this retains the performance edge over standard approaches.

Let $\phi^*_{\theta_\phi}$ and $K^*_{\theta_K}$ be neural networks with parameters $\pphi\in\R^{\nPP},~\pK\in\R^\nPK$, respectively: 
\[
    \begin{aligned}
        \phi^*_{\theta_\phi}&: \cX_a\rightarrow \R\\
        K^*_{\theta_K}&: \cT\times\cX\rightarrow\R^{\nU\times\nU}
    \end{aligned}
\]
which we use as candidate optimal desired energy and damping injection gain,
\[
    \begin{aligned}
        \phi^*(x_a) &\approx \phi^*_{\theta_\phi}(x_a),\\
        K^*(t, x) &\approx K^*_{\theta_K}(t, x).
    \end{aligned}
\]
Let $\theta:=(\pphi, \pK)\in\R^\nT~(\nT=\nPP+\nPK)$. Without loss of generality and for the sake of a compact notation we simply write $\phi^*_{\theta},~K^*_{\theta}$ and $u_\theta$. The finite--dimensional optimization problem can be finally rewritten as
\begin{equation}\label{eq:NPunc2}
        \begin{aligned}
            \min_{\Th}  & \mathbb{E}_{x_0\sim\bP(x_0)}\left[\ell(x_0, u_\Th)\right]\\
            \text{subject to}~~
            &\dot x =  F(x)\nabla\Ha(x) + g(x)u_\Th(t,x)~~t\in [0,T]\\
            & u_\Th(t, x) = -g^+(x)F^\top(x)\nabla_x\phi^*_\Th(x_a)\\
            &\qquad~~~~- K^*_\Th(t, x)g^\top\nabla\Ha(x)\\
            &K^*_\Th(t, x)\succ 0
        \end{aligned}
\end{equation}
%
%

%
%
Then, the optimization problem turns in finding via SGD the optimal parameters $\theta$ by iterating
\begin{equation}\label{eq:SGD}
        \theta_{k+1} = \theta_{k} - \frac{\eta_k}{N}\sum\limits_{i=1}^N\dfrac{\dd}{\dd\theta}\ell(x_0^i, u)
\end{equation}
\begin{rmk}
    If strong priors over the functional class of $\phi$ and $K^\star$ are given, the approximation problem can be rewritten onto an orthonormal basis where the functions are realized trhough a truncated eigenfunction expansion in which we optimize for the eigenvalues.
\end{rmk}

In order to perform the above gradient descent we therefore need to compute the gradient (i.e. the sensitivity) of $\ell(x_0, u)$ with respect to the parameters $\theta$. 
\subsection{Computing the Sensitivities}
In order to compute the the quantity $\dd \ell/\dd \theta$ let us recall some recent results of the machine learning community developed in the context of Neural ODEs \cite{chen2018neural}. For the sake of compactness let us denote with $f_\theta(t, x)$ the port-Hamiltonian vector field controlled via energy shaping, i.e. 
$$
    f_\theta(t, x) :=F(x)\nabla\Ha(x) + g(x)u_\theta(t, x).
$$
In order to compute $\dd \ell(x_0^i, u)/\dd \theta$, the initial value problem 
$$ \dot x = f_\theta(t, x),~~x(0) = x_0^i $$
is solved within the time horizon $[0, T]$ via a chosen (possibly adaptive--step) ODE solver. The obtained solution is then used to evaluate the cost $\ell$. From a software implementation perspective $\dd \ell/\dd \theta$ could be approximated by taking advantage of some suite of \textit{automatic differentiation} \cite{baydin2017automatic} (AD) available in most machine--learning--related libraries (e.g. {\tt Pytorch} \cite{paszke2019pytorch}, {\tt TensorFlow} \cite{abadi2016tensorflow}, {\tt JAX} \cite{jax2018github} or {\tt Flux} \cite{innes2018fashionable}) . However, such techniques introduce numerical errors in the gradients and, moreover, in several cases they can hardly be applied in practice due to their cumbersome memory footprint. As discussed in \cite{chen2018neural}, if we suppose that the numerical solver takes $\tilde{T}$ steps to integrate the ODE in $[0, T]$, both the computational complexity and memory efficiency of AD gradients result to be $\cO(\tilde{T})$. Therefore, even for small-sized neural networks $\phi_\theta^*$, $K^*_\theta$ if the ODE is solved with high accuracy and long integration times (and, consequently, a very large $\tilde T$) the memory required to compute the AD gradients grows rapidly. 
For this reason, the \textit{adjoint method} for gradient computation has been introduced in the field, considerably increasing the memory efficiency to $\cO(1)$. This approach made accessible the training of large Neural ODEs with up to millions of parameters for different deep learning tasks \cite{massaroli2020dissecting, zhuang2020adaptive, grathwohl2018ffjord}, practically impossible using AD. The adjoint method, based on Pontryagin maximum principle \cite{pontryagin1962mathematical}, has been introduced in the context of Neural ODEs by \cite{chen2018neural} and generalized in \cite{massaroli2020dissecting} to account for integral cost functions, parameter-dependent costs as well as time-dependent parameters. In our case we will simply use the following result:
\begin{restatable}[Generalized Adjoint Method {\cite[Theorem 1]{massaroli2020dissecting}}]{theorem}{Adjoint}\label{thm:adjoint}
    Consider the initial value problem $\dot x = f_\theta(t, x), ~x(0)=x_0$ and the cost function 
    $$
    \begin{aligned}
        \ell(x_0) &= \int_0^T l\left(\tau, \theta,\Phi_\tau(x_0, u)\right)\dd\tau \\
        &+ L\left(T, \theta, \Phi_T(x_0, u)\right).
    \end{aligned}
    $$
    Then $\dd\ell/\dd{\theta}$ is given by 
    \begin{equation}
        \frac{\dd\ell}{\dd\theta} = \frac{\partial L}{\partial\theta} + \int_0^T\left[\lambda^\top
        \frac{\partial f_\theta(t, x)}{\partial\theta} + \frac{\partial l}{\partial\theta}\right]\dd\tau
    \end{equation}
    where the Lagrangian multiplier $\lambda(t)$ satisfies
    \begin{equation*}
       \dot{\lambda}^\top = -\lambda^\top\frac{\partial f_\theta(t, x)}{\partial x}- \frac{\partial l}{\partial x},~~
       \lambda^\top(T) = \frac{\partial L}{\partial x(T)}
    \end{equation*}
    solved backward in $[0, T]$.
\end{restatable}
Note that, the constant memory efficiency of the adjoint method is a consequence of not needing to store the solution of the system along the trajectory in order to implement some finite--difference--based method (such as AD). Indeed, these methods only require a solution of the adjoint ODE backward in time. In this case the obtained gradient are \textit{exact} if the the adjoint ODE is solved correctly. Numerical errors which arise due to the use of ODE solvers can be bounded by choosing suitable solver tolerances. 

\subsection{Set stability of the optimal energy minima}
As a common practice in robotics and passivity--based control in general, stability is always guaranteed for lower--bounded energy functions. In fact, in presence of friction or a damping injection control action, the system will dissipate all its energy while the state converges towards a minimum point of $\Ha(x)$. However, when the shaped energy is the output of the optimization process \eqref{eq:NPunc2} via stochastic gradient descent \eqref{eq:SGD}, how can we guarantee lower--boundedness of the optimal shaped energy without an explicit nonlinear constraint posed in the optimization phase? Since $\Ha$ is lower bounded, this can be achieved imposing boundedness of $\phi^*_\theta$. When using neural networks this can be easily obtained by using $\gamma$--bounded (scalar) activation functions $\sigma:\R\rightarrow\R$, i.e. such that $\exists \gamma>0 : \forall \alpha\in\R~|\sigma(\alpha)|\leq \gamma$.

Notable examples of such class of activations are the \textit{sigmoid} functions $\sigma(x) := e^x/(1+e^x)$ or hyperbolic tangent $\sigma(x) := \tanh{(x)}$. 
The following result proves that, shallow neural networks with $\gamma$--bounded activation\footnote{Note that, following the machine learning convention, activation functions are though to be acting element--wise, i.e. $\sigma(x) = [\sigma(x_1), \dots, \sigma(x_n)]$.} have bounded output.
\begin{lemma}\label{lem:2}
    Let $f_\theta(x)$ be a shallow neural network
    \[
        f_\theta(x) = w^\top\sigma\left(Wx + b\right),\quad \theta:=\vect{(w,W,b)}\in\R^{(2+\nX)\nH}
    \]
    with $w,b\in\R^\nH$, $W\in\R^{\nH\times\nX}$ and $\sigma(\cdot)$ is $\gamma$--bounded. Then
    \[
        \forall x\in\R^\nX~~|f_\theta(x)|\leq\gamma\sum_i|w^{(i)}|
    \]
    \begin{proof}
        \[
            f_\theta(x) = \sum_iw^{(i)}\sigma\underbrace{\left(\sum_jW^{(i,j)}x^{(j)} + b^{(j)}\right)}_{z^{(i)}(x)}
        \]
        Since, for any $x$, $\left|\sigma(z^{(i)}(x))\right|\leq \gamma$ then, $|w^{(i)}\sigma(z^{(i)}(x))|\leq \gamma|w^{(i)}|$. Therefore, by the triangle rule,
        \[
            \begin{aligned}
                |f_\theta(x)| &=\left|\sum_iw^{(i)}\sigma(z^{(i)}(x))\right|\\
                &\leq \sum_i\left|w^{(i)}\sigma(z^{(i)}(x))\right|\leq\gamma\sum_i|w^{(i)}|
            \end{aligned}
        \]
    \end{proof}
\end{lemma}
Note that if $f_\theta$ was a deep neural network, Lemma \ref{lem:2} would still hold. In fact, even if $z$ was the output of a deep network:
\[
    z(x) : = W_N\sigma(W_{N-1}(\sigma(\cdots W_1x + b_1) + b_{N-1}) + b_N,
\]
it holds $|\sigma(z^{(i)}(x))|\leq \gamma$ for any $i$. This also shows that, in principle, it is sufficient to require $\gamma$--bounded activations in the last layer only.

Similarly, positive--definiteness of $K^*_\theta$, i.e. $\forall t~~ K^*_\theta(t, x(t))\succeq 0$ must be enforced in order to ensure that the damping injection controller is actually extracting energy from the system rather than replenishing it. If, e.g., $K^*_\theta$ is assumed to be diagonal, $K^*_\theta(t, x) = \diag(k^1_{\theta}(t, x), \cdots, k^\nU_{\theta}(t, x))$ where the elements $k^i_\theta(t, x)$ are the output of the neural network, positive definiteness may be easily obtained by appending a positive activation function, e.g. ReLU or \textit{SoftPlus}, to the output of the network. Note that, the choice of a diagonal $K^*_\theta$ does not constitute a loss of generality since every element $k^i_\theta$ is a function of both the state $x$ and time $t$. Thus, cross--correlations between different outputs which may arise from non--diagonal matrices $K^*_\theta$ could be indirectly learned inside $k^i_\theta(t, x)$, if optimal.

Nonetheless, if $\phi^*_\theta$, $K^*_\theta$ are designed bounded and positive definite, respectively, the following proposition holds.
\begin{restatable}[Stability]{proposition}{stability}\label{prop:stab}
    Every closed set $\cM\subset\R^{\nX}$ such that
    \[
        \forall x\in\cM,~\nabla\Ha^*=\mathbb{0}_{\nX}~\text{and}~\nabla^2\Ha^*\succeq 0
    \]
    which is contained in an open neighborhood $\cB\supset\cM$ satisfying
    \[
        \forall x\in\cB\setminus\cM,~\nabla\Ha^*\neq\mathbb{0}_{\nX}
    \]
    is locally asymptotically stable. 
    %
\end{restatable}
\begin{proof}
    From Proposition \ref{prop:1}, it follows that the closed--loop system is still a port--Hamiltonian system with energy $\Ha^*$ and system matrix $F(x) - g(x)K^*(t, x)g^\top(x)$. Thus, for all $x\in\cM$, $\dot x = \mathbb{0}_{\nX}$ and, thus, $\cM$ is forward invariant. Let $\cB\supset\cM$ be an open neighborhood of $\cM$ such that $\forall x\in\cB\setminus\cM,~\nabla\Ha^*\neq\mathbb{0}_{\nX}$ and let $V(x) := \Ha^*(x)$ be a candidate Lyapunov function. 
    
    It holds:
    \begin{itemize}
        \item[$i)$~] $\forall x\in\cB$
        \[
            \begin{aligned}
            \dot V &= \langle \nabla\Ha(x), F(x)\nabla\Ha(x)\rangle \\
            &-\nabla\Ha(x)g(x)K^*_\theta(t, x)g^\top(x)\nabla\Ha(x)\leq 0
            \end{aligned}
        \]
        \item[$ii)$~] $\forall x\in\cM~~\dot V = 0$
        \item[$iii)$~] $\cM$ is the largest invariant subset of $\{x\in\cB:\dot V=0\}$
    \end{itemize}
    where Proposition \ref{prop:1} has been used in $i)$. Thus $\cM$ is asymptotically stable.
\end{proof}
This means that for all initial conditions $x_0\in\cX$ there will exist a set $\cM$ such that ${\displaystyle\lim_{t\rightarrow\infty}\Phi_t(x_0, u_\theta)\in\cM}$. 

In the next section we consider the optimal energy shaping problem for a robotic system. 
%
%

\section{Optimal Energy Shaping for a Robotic System}\label{sec:5}
In this paper, as a popular application domain for passivity--based techniques, we specialize on robotic systems which admit a port-Hamiltonian model in the form \eqref{eq:mech_ph}. For systems in such form, it is know that shaping only the potential component of the energy satisfies the matching conditions \eqref{eq:matching_cond}. This may be proven by extending Proposition \eqref{prop:PDE_general} as follows:
\begin{restatable}[Potential Shaping]{corollary}{PotShapCor}\label{cor:potential_shaping}
        Consider a \textit{fully--actuated} mechanical system (\ref{eq:mech_ph}). 
    Then, every desired energy
        \[\Ha^*(q,p) := \Ha(q,p) + V^*(q)\]
    is admissible for all $V^*\in \cC^\infty({\cQ\rightarrow\R})$, in the sense that it satisfies the matching equations (\ref{eq:matching_cond}).
\end{restatable}
\begin{rmk}
    In case of under--actuated systems it follows naturally that we can only partially shape the potential. In particular, let $q_a\in\R^a$ ($a<d$) be the generalized positions of the actuated degrees of freedom. Then, we can freely assign any desired energy function $\Ha^*(q,p) := \Ha(q,p) + V^*(q_a)$
\end{rmk}
In practice, most approaches (e.g. PD with gravity compensation) fully compensate the potential $V(q)$, and substitute it with another desired function $\bar{V}(q)$ , i.e.
\begin{equation}
    \begin{aligned}
        &V^*(q): = \bar V(q) - V(q) \\
        \Rightarrow &\Ha^*(q,p) = \frac{1}{2}p^\top M^{-1}(q) p + \bar V(q)
    \end{aligned}
\end{equation}
However, this line of work presents two major drawbacks regarding the \textit{compensation} of the system's potential:
\begin{itemize}
    \item Uncertainties on the model's parameters inevitably translates in partial miscompensation the potential, likely leading to the introduction of unstable nonlinear dynamics;
    \item A full compensation of $V(q)$ in favour of a quadratic potential is not necessary for most position--regulation tasks in the whole configuration manifold. These approaches are in fact, inconsistent with the paradigm of exploiting passive elements of the robot to reduce, e.g., the control effort. 
\end{itemize}
Instead, we should take advantage of the robot's potential (gravity, joints' elasticity) wherever they produce forces aiding the robot in its desired motion. In this work, we go beyond classic techniques by not compensating the potential and simply learning an added potential $V^*(q)$ such that the shaped energy becomes $\Ha(q, p) = p^\top M^{-1}(q)p / 2 + V(q) + V^*(q)$.

Using the same notation and arguments as in the general case developed before, we approximate the decision variable $V^*$ by a neural network $V^*_{\theta}$. Therefore, in this setting, the optimal energy shaping control problem \eqref{eq:NPunc2} translates to the robot case as follows:
\begin{equation}\label{eq:NPunc_robot2}{
        \begin{aligned}
            \min_{\theta}  & ~~\mathbb{E}_{(q_0, p_0)\sim \bP(q_0, p_0)}\left[\ell(q_0, p_0, u_\theta)\right] \\
            \text{s. t.}~~
            &
            \begin{bmatrix}
            \dot q\\
            \dot p
            \end{bmatrix} = 
            \begin{bmatrix}
                \mathbb{O} & \mathbb{I}\\
                -\mathbb{I} & D
            \end{bmatrix}
            \begin{bmatrix}
                \nabla_{q}\Ha(q, p)\\
                \nabla_{p}\Ha(q, p)
            \end{bmatrix} +
            \begin{bmatrix}
                \mathbb{O}\\
                B
            \end{bmatrix}u_\theta\\
            & t\in [0,T]\\
            & u_\theta = -B^{-1}\nabla_{q}V^*_\theta(q)- K^*_\theta(t, q, p)B\dot q\\
            & K^*_\theta(t, q, p)\succeq 0
        \end{aligned}}
\end{equation}

\section{Simulation Experiments}

We validate the proposed Optimal Energy Shaping (OES) control on a series of two different settings involving position regulation tasks of an inverted pendulum with elastic joint. In particular, we feature:
\begin{itemize}
    \item[]\fcircle[fill=deblue]{2pt} A direct comparison between the proposed novel optimal energy shaping control and the classic linear control (quadratic desired potential and linear damping) with potential compensation in which the proportional and derivative gains are tuned alongside the OES controller to minimize the same cost function. 
    \item[]\fcircle[fill=wine]{2pt} A \textit{generalization} task where we feed a desired position set point $q^*$, distributed according to a given probability density function, to the neural network parametrizing the optimal shaped potential, i.e. $V^*_\theta = V^*_\theta(q, q^*)$, 
 \end{itemize}
\subsection{Experimental Setup}
\begin{figure}
    \centering
    \includegraphics[scale=.95]{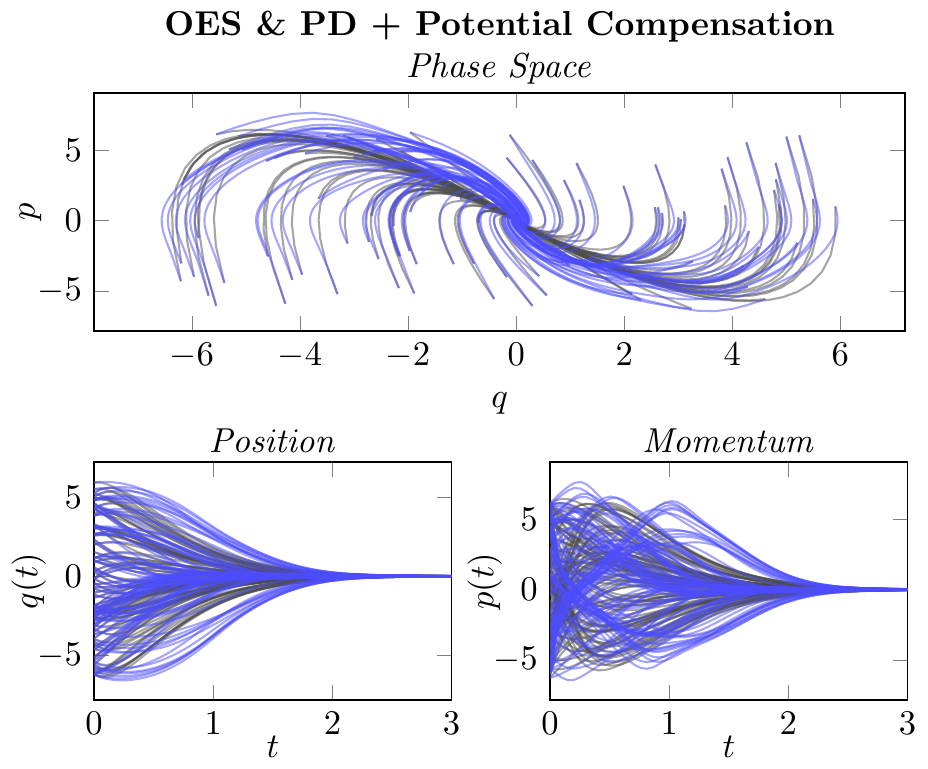}
    \caption{Closed--loop phase--space trajectories in the optimal energy shaping case ({\color{blue!70!white}blue} lines) compared to the classic PD with potential compensation ({\color{black!70!white}grey} lines) for different initial conditions sampled from the probability distribution of initial conditions $\phi_0$. The parameters of the two controllers were optimized to minimize the same cost function in a regulation task.}
    \label{fig:traj1}
\end{figure}
\begin{figure*}[t]
    \centering
    \includegraphics[scale=.95]{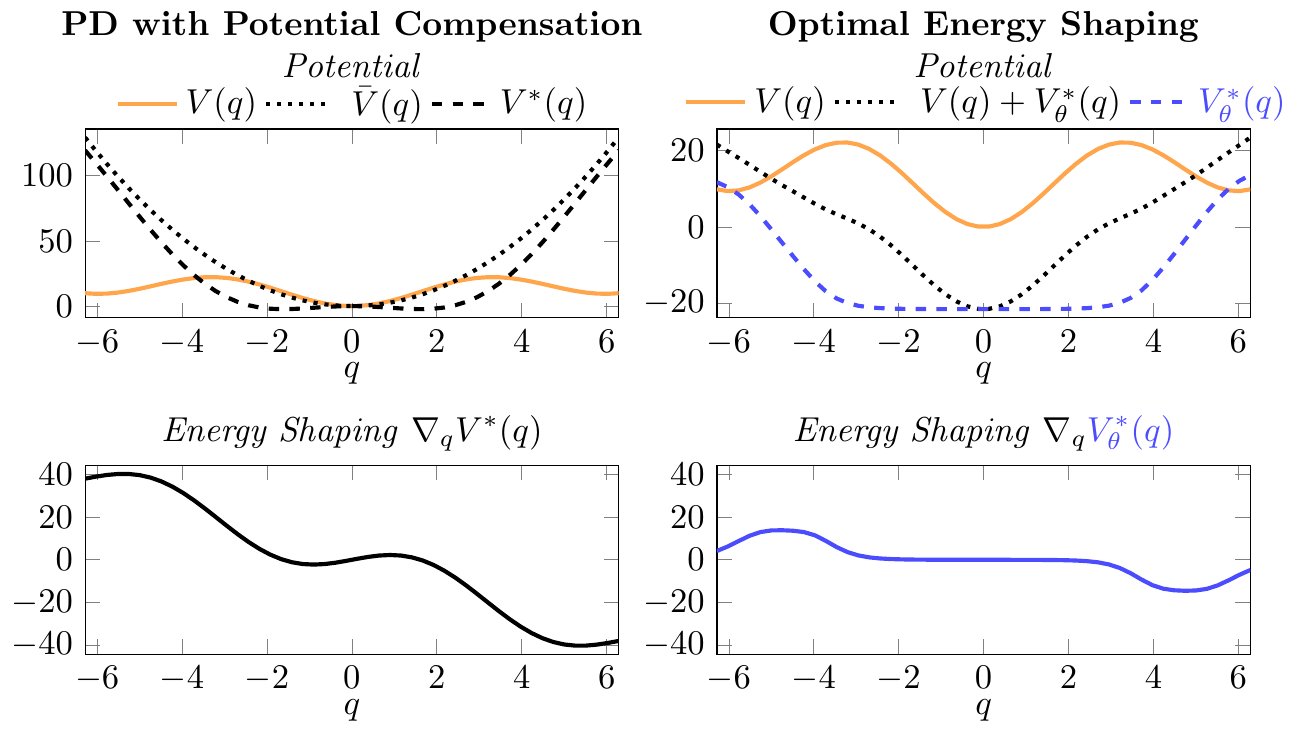}
    \caption{Energy shaping: A comparison between the PD with potential compensation and the OES approaches, both optimised for the same cost. It is clearly seen that the closed--loop potential $\Bar{V}$ in the PD case is by construction a parabola, while in the OES the major expressivity allowed by the neural network breaks this limitation and produces a nontrivial closed--loop potential $V+V_{\theta}^*$ which yields a better performance, as evident from the resulting control efforts.}
    \label{fig:exp1}
\end{figure*}
Let us first introduce the PH model of the inverted pendulum with elastic joint. The configuration manifold of the system is the one--sphere $\mathbb{S}^1$ and let $q\in\mathbb{S}^1$ be its position coordinate. The generalized momentum $p$ is then defined as $p=J\dot q$ where $J$ is the inertia of the pendulum. The total energy (Hamiltonian) of the system is 
$$
    \Ha(q,p) :=\frac{1}{2}J^{-1}p^2 + \underbrace{mgr(1-\cos q) + \frac{1}{2}kq^2}_{V(q)}.
$$
where $m,~r$ are respectively the mass and location of the center of mass of the pendulum. Moreover, $k$ is the stiffness of the torsional spring placed in the joint. By considering viscous friction with damping coefficient $\beta$ acting on the joint, the port-Hamiltonian realization of the model is 
$$
    \begin{aligned}
    \begin{bmatrix}
        \dot q\\
        \dot p
    \end{bmatrix}&=
    \begin{bmatrix}
    0 & 1\\
    -1& -\beta
    \end{bmatrix}\begin{bmatrix}
        \nabla_q \Ha \\
        \nabla_p \Ha
    \end{bmatrix} + \begin{bmatrix}
        0\\1
    \end{bmatrix}u\\
    y &= \dot q
    \end{aligned}
$$
In all experiments the model parameters have been fixed to $m=1~[Kg]$, $k=0.5~[N/\rad]$, $r=1~[m]$, $\beta=0.01~[Ns/\rad]$, $g=9.81~[m/s^2]$. 
The simulations have been carried out on a machine equipped with an \textsc{AMD Ryzen Threadripper 3960X} CPU and two \textsc{NVIDIA RTX 3090} graphic cards. The software has been implemented in {\tt Python} using {\tt torchdyn}~\cite{poli2020torchdyn} and {\tt torchdiffeq}~\cite{chen2018neural} libraries. In all the experiments, \textsc{Adam} optimizer \cite{kingma2014adam} has been used to perform the gradient descent iterations and, unless differently specified, its learning rate has been set to $10^{-3}$. After fixing a probability density function $\psi_0(x)$ of initial conditions, at each step of gradient descent, a batch of 2048 initial conditions $\{x_0^i\}_i=\{(q_0^i, p_0^i)\}_i$ were sampled from $\psi_0$ and integrated with a Dormand--Prince ({\tt dopri5}) ODE solver~\cite{dormand1980family} in a time horizon $t\in[0, T]$. Then, the cost function was computed and back--propagated via the adjoint method to obtain the cost gradient $\dd \ell/\dd \theta$. Finally the gradient descent step was carried out to update the parameters. All the models were trained with an integral cost chosen to be the control effort in the time horizon:
$$\ell(x_0) = L_\theta(x_0) + \gamma\int_0^T |u_\theta(t)|\dd t,\quad \gamma>0.$$
\paragraph*{\fcircle[fill=deblue]{3pt} Comparing OES to potential compensation}
First, we compared the optimal energy shaping to the classic \textit{PD with potential compensation controller} (PD+). In order to represent the regulation task in probabilistic form, we defined a target Normal probability density function $\psi^*(x)$ centered in the set point $(q^*,0)$ and variance $\sigma^2\mathbb{I}$. We then trained the model to maximize the likelihood of $\psi^*(x)$ at the terminal time $t=T$. This has been obtained by minimizing the approximated negative log--likelihood of the terminal state, i.e.
\begin{equation}\label{eq:exp_cost}
    \begin{aligned}
        \mathbb{E}_{x_0}[\ell(x_0)] \approx &-\frac{1}{N}\sum_{i=1}^N\log\psi^*\left(\Phi(x_0^i, u_\theta)\right)\\
        &+\frac{\gamma}{N}\sum_{i=1}^N\int_0^T \left|u_\theta\left(\Phi_t(x_0^i, u_\theta\right)\right|\dd t
    \end{aligned}
\end{equation}
\begin{figure*}
    \centering
    \includegraphics[scale=.95]{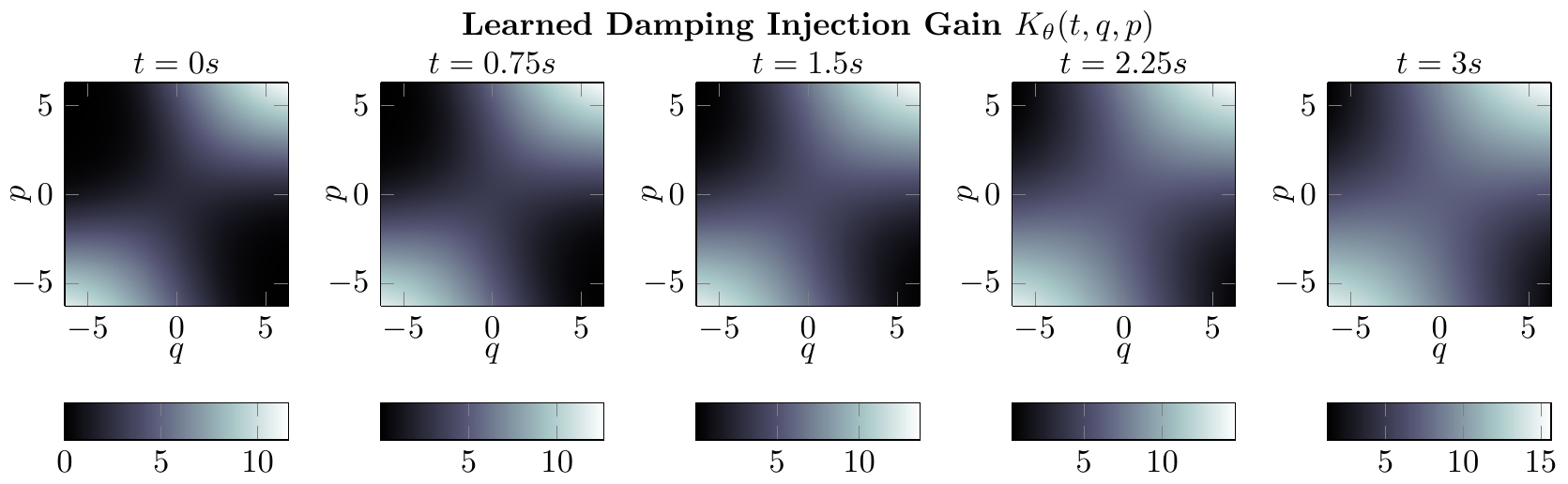}
    \caption{Time evolution of the state--dependent learned damping injection gain $K_\theta(t, q, p)$. Notice that at each plot exhibits a nontrivial distribution of the damping injection gains in the state space, greatly generalising the constant derivative action characteristic of the PD with potential compensation implementation. As further generalisation we implemented a smooth dependence on time directly, showing that a feedforward component of the control action would increase the performance along the task in nominal conditions.}
    \label{fig:damping_injection_gain}
\end{figure*}
In particular, we chose $q^*=0$, $\sigma^2=10^{-3}$, $T=3$. In order to obtain a fair comparison between the proposed optimal energy shaping and the PD with potential compensation controllers, we decided to optimize as well the hyper--parameters of the latter. Given the PD with potential compensation law
$$
    u_{\theta}(q, p) = \overbrace{-mgr\sin q - kq}^{\text{Potential Compensation}} + \underbrace{{\color{blue!70!white}k_q}(q - {q^*})}_{\nabla_q\Bar{V}(q)} - {\color{blue!70!white}k_d}J^{-1}p 
$$
we optimized the parameters $\theta=(k_p, k_d)$ with the cost function \eqref{eq:exp_cost}. The OES controller has been instead designed following \eqref{eq:NPunc_robot2}:
\begin{equation}
    u_\theta(q, p) =-\nabla_{q}{\color{blue!70!white}V^*_\theta(q)}- {\color{blue!70!white} K^*_\theta(t, q, p)}J^{-1}p.
\end{equation}
$V^*_\theta:\mathbb{S}^1\rightarrow\R$ was realized by a multi--layer perceptron comprising two hidden layers of $64$ neurons each with ${\tt SoftPlus}:x\mapsto \log(1+e^x)$ and $\tanh$ activations on the two hidden layer, respectively. $K^*_\theta(t, q, p)$ was build with a single hidden layer of 64 neurons with {$\tt SoftPlus$} activation on the hidden layer as well as appended to the output of the Neural Network.
\begin{figure}[b]
    \centering
    \includegraphics[scale=.95]{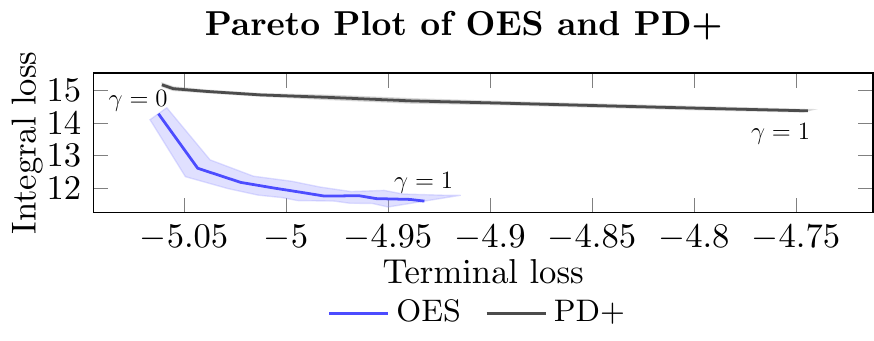}
    \caption{Pareto plot of the terminal and integral cost functions for different values of the weight $\gamma$, ranging from zero to one. The resulting pareto front of the proposed OES approach always lies closer to $[-\infty, 0]$ than the one of the PD+, confirming OES's higher degree of optimality in the considered task.}
    \label{fig:loss_pareto}
\end{figure}
We trained both the OES and the PD controllers on the cost \eqref{eq:exp_cost} with $\gamma=0.01$ and learning rates $\eta_{\tt OES}=10^{-3}$ and $\eta_{\tt PD} = 1$, respectively. All the parameters $\theta$ of the neural networks $V_\theta^*$, $K_\theta^*$ have been initialized according the Xavier scheme \cite{glorot2010understanding}. In line with the some recent heuristics regarding the initialization Neural ODEs parameters \cite{kidger2020neural}, the parameters of the last layers of both $V^*_\theta$ and $K^*_\theta$ have been then set to zero, i.e., at first step of optimization process, the controller will be always zero. Similarly in the PD+ case, $k_p$ and $k_d$ have been, initialized to zero.
 
 The probability density function $\psi_0(x)$ from which samples of initial conditions were drown was chosen as an uniform distribution in $[-2\pi, 2\pi]\times[-2\pi, 2\pi]$. Conversely, the variance of the target distribution $\psi^*$ was set as $\sigma^2 = 10^{-3}$. The optimized parameters of the PD controller converged to $k_p=6.55,~ k_d = 4.89 $. The values of the cost functions at convergence for both the OES and the PD with compensation controller are reported in Table \ref{tab:loss}. The resulting closed--loop phase--space trajectories of both methods are displayed in Figure~\ref{fig:traj1}.
\begin{table}[b]
    \centering
    \begin{tabular}{r|c|c}\toprule
               Cost Function       & OES & PD+  \\\midrule
     Terminal~~~$-\frac{1}{N}\sum_{i=1}^N\log\psi^*\left(\Phi(x_0^i, u_\theta)\right)$  & \bf{-5.00} & -4.89\\
         Integral $\frac{1}{N}\sum_{i=1}^N\int_0^T \left|u_\theta\left(\Phi_t(x_0^i, u_\theta\right)\right|\dd t$ & \bf{13.28} & 15.73\\\bottomrule
    \end{tabular}
    \vspace{3mm}
    \caption{Terminal and Integral Costs at convergence computed on samples $\{x_0^i\}_i\sim\psi_0$ with $\gamma=0.01$.}
    \label{tab:loss}
\end{table}

Figure \ref{fig:exp1} shows that the energy shaping control input learned by the system takes advantage of gravity to perform the task with a considerably lower control effort than the PD case. In particular, around the set point, the learned potential is such that the control action, and thus its effort, basically vanishes, letting the gravity and the damping injection do the job. Figure \ref{fig:damping_injection_gain} shows that the learned damping injection coefficient, as function of the system state, interestingly assumes values which are not symmetric with respect to the $p=0$ axis in the phase space. Contrarily to the constant damping characteristic of the PD implementation, the system damps an amount of energy which depends both on the sign of the velocity (which determines e.g. whether gravity is computing positive or negative work on the system) and the position, which determines which potential functions are dominant in that state. This represents a major generalisation of energy shaping and damping injection with respect to its implementation in the PD with potential compensation form, since the control action is free to range on a much bigger decision space, while remaining passive and thus conserving the advantages of the base methodology. As further generalisation we let the damping injection gains directly depend on time as well, leading to a feed--forward component of the control action which could increase performance in known environmental conditions. 
Due to the multi--objective nature of the optimization problem, we further evaluated the effect of the cost weighting coefficient $\gamma$ in an ablation study, varying its value from zero to one in ten equally--spaced steps. This has been performed through Monte Carlo simulation of 25 runs for each value of $\gamma$ with different initialization of the parameters of the neural networks $V_\theta^*$ and $K_\theta^*$ as well as the coefficients of the PD controller. The obtained pareto front is displayed in Figure~\ref{fig:loss_pareto}. The overall results clearly show that the OES methodology leads to a nontrivial, passive, nonlinear control action which yields a better performance with respect to the PD with potential compensation. This simple example is insightful because the results can be interpreted easily. 
\paragraph*{\fcircle[fill=wine]{3pt} Generalizing to arbitrary set points}
We further evaluate the OES approach by feeding a desired position set point to both the learned potential $V^*_\theta:(q,q^*)\mapsto V_\theta^*(q, q^*)$ and damping injection gain $K^*_\theta:(q,p, q^*)\mapsto K^*_\theta(q, p,q^*)$. In particular, we aimed at constructing a single controller able to generalize across different $q^*$. In such case we consider a distribution over target positions and momenta $x^*=(q^*,p^*)$ and subsequently constructed the terminal cost function as
$$
    L(x(T), x^*) := \mathbb{E}_{x^*}\left[(x(T)-x^*)^\top Q(x(T)-x^*)\right]
$$
with $Q = Q^\top>0$ resulting in the overall cost
$$
    \begin{aligned}
        \mathbb{E}_{x_0}[\ell(x_0)] = \mathbb{E}_{x_0}&\bigg[\mathbb{E}_{x^*}\left[e^\top_\theta(x_0, x^*) Qe_\theta(x_0, x^*)\right]\\
        &+\frac{1}{2}\int_0^1\left[u_\theta\left(\Phi_t(x_0, u_\theta\right)\right]^2\dd t\bigg]
    \end{aligned}
$$
with $e_\theta(x_0, x^*):=\Phi_{T}(x_0, u_\theta)-x^*$.
In practice, at each gradient descent iteration, we also sampled target positions and momenta $q^*, p^*$ uniformly in $[-2\pi,2\pi]\times[-10^{-4},10^{-4}]$ for each initial condition $x_0$. We thus obtained the approximate expectation of the terminal cost by 
$$
    \begin{aligned}
        \mathbb{E}_{x_0}[L] \approx \frac{1}{NM}\sum_{i=1}^N \sum_{j=1}^M e^\top_\theta(x_0, x^*)  Q e_\theta(x_0, x^*) 
    \end{aligned}
$$
where $Q$ was picked as $\diag([10,1])$.
Here, the Neural Networks parametrizing $V^*_\theta$ and $K^*_\theta$ have been chosen similar to the previous ones: an additional ${\tt SoftPlus}$--activated hidden layer has been added to both networks while the number of units in the hidden layers has been increased from 64 to 128. The controller was trained on a time horizon $[0, 1]$.
\begin{figure}[t]
    \centering
    \includegraphics[width=\linewidth]{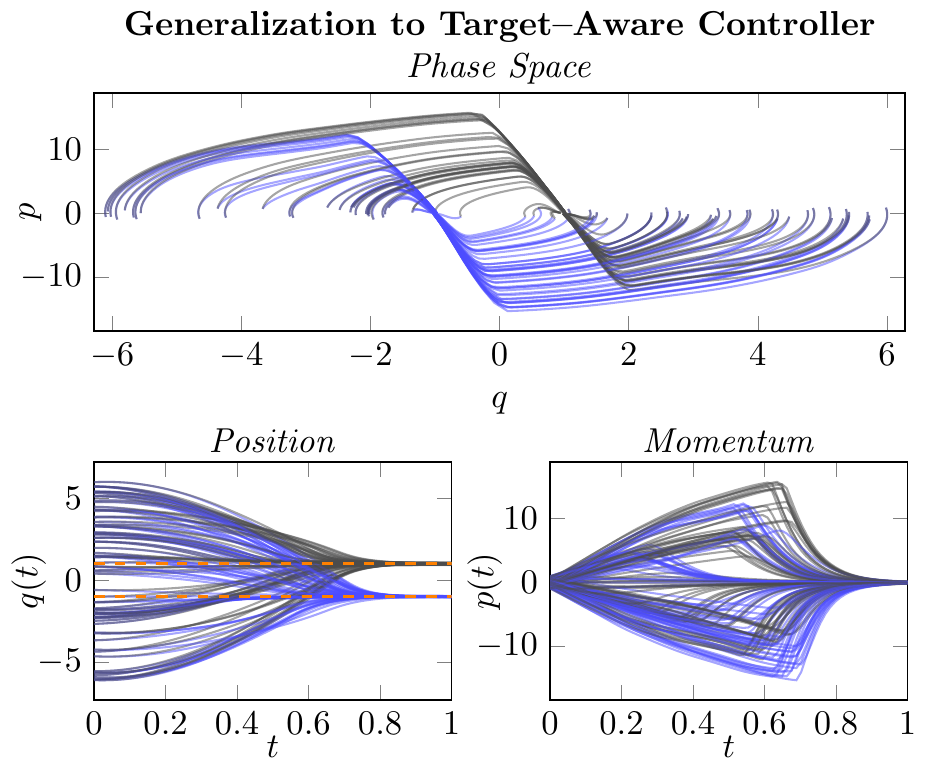}
    \vspace{-5mm}
    \caption{Phase--space trajectories of the optimal energy shaping controlled system trained on arbitrary set points $q^*$. Trajectories starting from the same (random) initial conditions asymptotically converge to different set points ({\color{blue!70!white} $q^*=-1$}, {\color{black!70!white} $q^*=1$}).}
    \label{fig:exp2}
\end{figure}
\begin{figure*}[t]
    \centering
    \includegraphics[scale=.95]{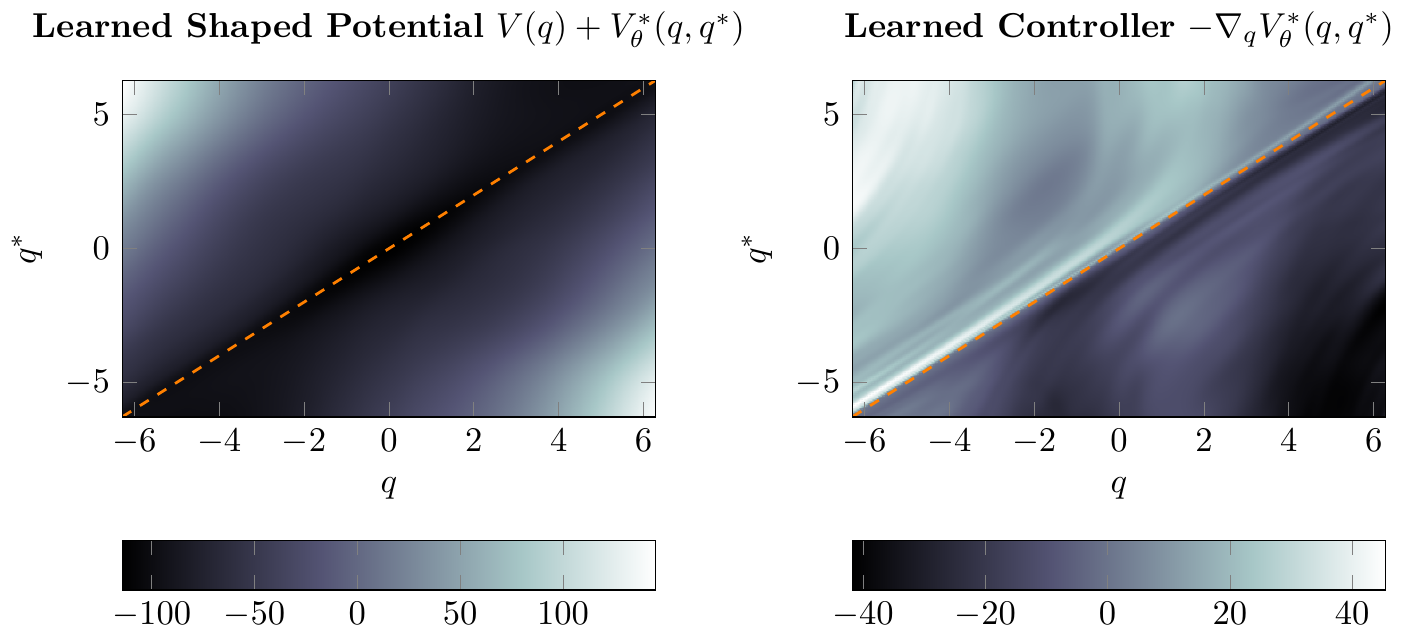}
    \caption{Learned optimal potential trained on arbitrary set points $q^*$ and corresponding energy shaping control action. The learned potential $V(q) + V_\theta^*(q, q^*)$ presents its minima on the line $q=q^*$, guaranteeing asymptotic stability of the desired configurations. On the other hand, unlike the classic quadratic potential, $V(q)+V^*(q, q^*)$ shows highly nonlinear behaviors around the minima which translate in a nonlinear control action $-\nabla_{q}V^*_\theta(q)$.}
    \label{fig:exp2ES}
\end{figure*}
\begin{figure*}[t]
    \centering
    \includegraphics[scale=.95]{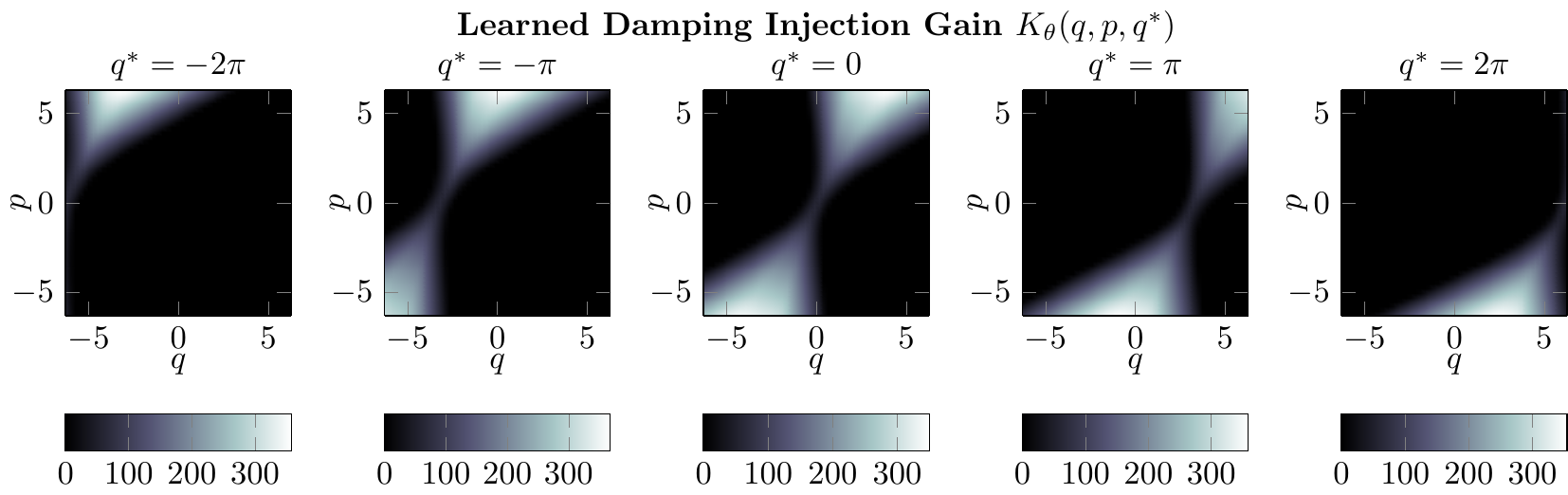}
    \caption{Learned state--dependent damping injection gains for different values of $q^*$. The learned gain $K_\theta(q, p, q^*)$ drastically change for the different values of $q^*$.}
    \label{fig:exp2DI}
\end{figure*}
{Figure \ref{fig:exp2} shows how the learned controller makes the system properly converge to different set points at the same rate. The graphs depicting the closed--loop potential and the energy shaping control action on an extended domain in which the set point $q^*$ spans an axis is shown in Figure \ref{fig:exp2ES}. Notice the minima on the line $q=q^*$, guaranteeing stability for any fixed $q^*$, and the nontrivial emerging control action, smoothly texturing the landscape visible by means of the used color code. In Figure \ref{fig:exp2DI} one can appreciate the further dependence of the damping injection gains on different slices of constant $q^*$, yielding an analogous interpretation to the one described in the previous experiment.}
%

\section{Related Work}
Other relevant works attempting to shed light on the link between port-Hamiltonian modeling, energy based methods and optimal/learning control can be found in the excellent survey paper \cite{nageshrao2015port} (and references therein), where adaptive and iterative learning approach are discussed. Moreover, \cite{zhong2020benchmarking} provides a comprehensive comparison of different energy--based machine learning approaches. In \cite{sprangers2014reinforcement} \textit{reinforcement learning} (RL) techniques are applied to port-Hamiltonian systems. Alternatives to traditional PBC can be found in \cite{Klsch2020OptimalCO}, where time--continuous feedback controllers for port--Hamiltonian systems are designed using adaptive dynamic programming principles. An additional connection between the Port--Hamiltonian framework and machine learning consists in exploiting the structure of dissipative PH systems to design continuous--time optimizers for neural networks \cite{massaroli2019port}. 

Recent work on port--Hamiltonian systems in machine learning \cite{zhong2019symplectic, zhong2020dissipative} explored a direct to learn mechanical port--Hamiltonian dynamics from data by learning all the components of the model (e.g. Hamiltonian Function, Inertia Matrix, Dissipation, etc.). The latter and recent results on unsupervised learning of Lagrangian systems \cite{zhong2020unsupervised}, pave the way to the extension of the proposed Optimal Energy Shaping for end--to--end learning of physically--consistent control policies directly from sensor data.

The methodology presented in this work greatly differs from data--driven approaches to decision making such as deep RL \cite{arulkumaran2017deep}. Here, we focus on providing a systematic blending of deep learning function approximators and the PBC paradigm. We discuss how this preserves auspicable theoretical properties of the control framework whilst improving its task--based performance. We highlight similarities with the introduction of expressive neural network function approximators in classical RL \cite{bertsekas1995dynamic,sutton2018reinforcement}, which allowed its principles to be applied in settings previously though intractable due to the curse of dimensionality. The goal of this work is to showcase a similar phenomenon for the PBC framework.  
\section{Conclusions and Future Work}
In this work we introduced the control paradigm of optimal energy balancing for robotic systems. At its core, this approach blends passivity based control, optimal control and neural--networks. Unlike traditional passivity based control, tasked with assigning quadratic energy functions and fictitious linear damping terms to stabilize some set point, \textit{optimal energy shaping} casts control synthesis in an optimal control setting, solved with fine machine learning techniques while implicitly meeting stability and passivity requirements.

Our approach paves the way to an holistic approach to robot control, focused on \textit{shaping} the robot's behavior in assigned tasks rather than to merely achieve stabilization or trajectory tracking. In this perspective, future work will include the extension to geometric control in the robot's work--space, task design (with e.g. \textit{reinforcement learning methods}) and improved optimization techniques. 

\bibliographystyle{IEEEtran}
\bibliography{main}

\appendix
\section{Proofs}
\subsection{Proof of Proposition \ref{prop:PDE_general}}
\begin{proof}
    Recalling the matching PDE (\ref{eq:matching_cond}), we have:
    \[
        \begin{aligned}
            &\begin{bmatrix}
                g^\perp(x) F^\top(x)\\
                g^\top(x)
            \end{bmatrix}
            \left(
                \nabla\Ha^*(x) - \nabla\Ha(x) 
            \right)
            = \mathbb{0}\\
            \Leftrightarrow & \nabla\Ha^*(x) - \nabla\Ha(x) \in \ker\left(
                \begin{bmatrix}
                    g^\perp(x) F^\top(x)\\
                    g^\top(x)
                \end{bmatrix}
            \right)
        \end{aligned}
    \]
    Thus, there exists a matrix--valued map $\Lambda:\cX\rightarrow\R^{n_a \times n_a}$ such that
    \begin{equation*}
        \ker\left(
            \begin{bmatrix}
                g^\perp(x) F^\top(x)\\
                g^\top(x)
            \end{bmatrix}\right)
        =
        \Span\left(\begin{bmatrix}
                \Lambda(x)\\
                \mathbb{O}
            \end{bmatrix}\right)
    \end{equation*}
    and, consequently, it holds,
    \begin{equation*}
         \nabla\Ha^*(x) - \nabla\Ha(x) 
        =
        \begin{bmatrix}
            \Lambda(x)\\
            \mathbb{O}
        \end{bmatrix}\alpha
    \end{equation*}
    for some $\alpha\in\R^{n_a}$. We can write 
    \begin{equation*}
        \begin{bmatrix}
             \nabla_{x_a}\Ha^*(x) - \nabla_{x_a}\Ha(x) \\
             \nabla_{x_b}\Ha^*(x) - \nabla_{x_b}\Ha(x)
        \end{bmatrix}
        =
        \begin{bmatrix}
            \Lambda(x)\\
            \mathbb{O}
        \end{bmatrix}\alpha,\quad x:=(x_a,x_b)
    \end{equation*}
    Moreover, if $\Lambda(x)$ is full--rank ($\forall x\in\cX~\rank(\Lambda)=n_a$), for any $\Ha^*(x)$ there exists an $\alpha$ such that $\nabla_{x_a}\Ha^*(x) - \nabla_{x_a}\Ha(x) = \Lambda(x)\alpha$. Therefore, the matching PDE reduces to
    \begin{equation*}
        \begin{aligned}
            &\nabla_{x_b}\Ha^*(x) - \nabla_{x_b}\Ha(x) = \mathbb{0}
        \end{aligned}
    \end{equation*}
    which is solved by any $\Ha^*(x)$ such that
    \[
        \Ha^*(x) = \Ha(x) + \phi^*(x_a) + c \quad (c\in\R)
    \]
    with $\phi^*:\R^\nA\rightarrow\R$, proving the result. 
\end{proof}
\subsection{Proof of Corollary \ref{cor:potential_shaping}}
%
\begin{proof}
    We have 
    \[
        g = 
        \begin{bmatrix}
        \mathbb{O}\\
        B
        \end{bmatrix}
        \Rightarrow \forall \tilde{B}\in\R^{n_q\times n_q},~g^\perp = \begin{bmatrix}
        \tilde{B} & \mathbb{O}
        \end{bmatrix}
    \]
    Thus, the matching condition becomes 
    \[
        \begin{aligned}
            &\begin{bmatrix}
                \begin{bmatrix}
                    \tilde{B} & \mathbb{O}
                \end{bmatrix}
                \begin{bmatrix}
                             O & -\mathbb{I}\\
                    \mathbb{I} & D
                \end{bmatrix}\\
                \begin{bmatrix}
                    \mathbb{O} & B 
                \end{bmatrix}
            \end{bmatrix}\left(\nabla\Ha^*(q, p) - \nabla\Ha(q,p)\right) =\mathbb{0}\\
            \Leftrightarrow&\begin{bmatrix}
                \mathbb{O} & \tilde{B}\\ 
                \mathbb{O} & B 
            \end{bmatrix}\left(\nabla\Ha^*(q, p) - \nabla\Ha(q,p)\right)=\mathbb{0}
        \end{aligned}
    \]
    Indeed, there exists a full--rank $n_q$--by--$n_q$ matrix $\Lambda$ such that
    \[
        \ker\left(\begin{bmatrix}
                \mathbb{O} & \tilde{B}\\ 
                \mathbb{O} & B 
            \end{bmatrix}\right) = \Span\left(\begin{bmatrix}
                \Lambda\\ 
                \mathbb{O} 
            \end{bmatrix}\right)
    \]
    Thus, for Proposition \ref{prop:PDE_general} the PDE becomes
    \[
        \nabla_p\Ha^*(q, p) - \nabla_p\Ha(q, p) = \mathbb{0}
    \]
    and solved by any desired energy whose additive terms solely depend on $q$, $\Ha^*(q,p) := \Ha(q,p) + V^*(q)$, $V^*\in\cC^\infty(\cQ\rightarrow\R)$.
\end{proof}
\end{document}